\documentclass[12pt,a4paper]{article}
\usepackage{amssymb,amscd,amsmath,amsthm}
\usepackage[colorinlistoftodos]{todonotes}
\usepackage[colorlinks=true, allcolors=blue]{hyperref}
\usepackage{lmodern}
\usepackage[english]{babel}
\usepackage{color}
\usepackage{graphicx}
\usepackage{tikz}
 \usetikzlibrary{shapes,arrows,fit}

 \usepackage{algpseudocode}
 \usepackage{algorithmicx}
 \usepackage{algorithm}
 \usepackage[lmargin=0.7in,rmargin=0.7in,tmargin=0.65in,bmargin=0.5in,footskip=0.25in]{geometry}

\usetikzlibrary{positioning, calc, arrows.meta}
\newtheorem{theorem}{Theorem}[section]

\newtheorem{corollary}[theorem]{Corollary}
\newtheorem{lemma}[theorem]{Lemma}

\newtheorem{remark}[theorem]{Remark}

\begin{document}

\title{Cocycle Actions on Hidden Quantum Markov Models: Symmetry Protection and Topological Order}

\author{
    Abdessatar Souissi \\
    \small $^1$ Department of Management Information Systems, College of Business and Economics \\
    \small Qassim University, Buraydah 51452, Saudi Arabia \\
\small $^2$ Institut Préparatoire aux Etudes des Ingénieurs de Monastir, Tunisia.
    \small \texttt{a.souaissi@qu.edu.sa}
    \and
    Abdessatar Barhoumi \\
    \small Department of Mathematics and Statistics, College of Science \\
    \small King Faisal University, Al-Ahsa PO.Box: 400, 31982, Saudi Arabia \\
    \small \texttt{abarhoumi@kfu.edu.sa}
}
\vskip0.5cm

\date{}
\maketitle

\begin{abstract}
We develop a symmetry action framework for hidden quantum Markov models (HQMMs) tailored to one-dimensional quantum spin systems and symmetry-protected topological (SPT) phases. In our setting, a symmetry group $G$ acts projectively on the hidden (virtual) degrees of freedom and linearly on the physical observation space, yielding a global HQMM state that is invariant under the combined action of $G$ for both conventional and causal (input--output) structures. We show that such symmetry actions are naturally classified by a group-cohomology $2$-cocycle $[\omega] \in H^{2}(G,\mathrm{U}(1))$, in direct analogy with the standard cohomological classification of one-dimensional bosonic SPT phases via projective edge representations. As an explicit example, we apply this construction to the Affleck--Kennedy--Lieb--Tasaki (AKLT) chain, where the hidden layer carries a nontrivial class $[\omega] \in H^{2}(\mathrm{SO}(3),\mathrm{U}(1))$ encoding its SPT order. In this case the HQMM formalism reproduces the known SPT properties of the AKLT state while providing a stochastic, Markovian description of the underlying virtual dynamics. Our results establish HQMMs as a natural bridge between quantum stochastic processes, tensor-network descriptions of many-body systems, and symmetry-protected topological order.
\end{abstract}
\vspace{0.3cm}


\vspace{0.2cm}


\section{Introduction}
Symmetry-protected topological (SPT) phases provide a unifying framework for a broad class of quantum spin systems in which global symmetries protect robust edge modes, quantized topological invariants, and nontrivial entanglement structures, even in one-dimensional settings \cite{H83,Pollmann2010,Pollmann2012}. Within this picture, the Haldane phase of spin‑1 chains and its paradigmatic representative, the AKLT model, furnish canonical examples of bosonic SPT phases protected by spin-rotation and related symmetries \cite{H83,AKLT1987,Affleck1988}. Building on the index-theoretic and group‑cohomological classification of 1D SPT phases, these systems are now understood in terms of projective edge representations and cohomology classes such as \([ \omega ] \in H^{2}(G,\mathrm{U}(1))\), and they are increasingly viewed as resources for quantum information processing and computational power \cite{O21,O23,Tas23,Chen2011a,Chen2011b,D17}. Recent developments have extended this perspective to open and driven quantum systems, where SPT order persists under symmetry-preserving noisy dynamics and gives rise to refined notions of order and protected structure in non‑equilibrium settings \cite{GTS22,D20,Sati26}.  

In parallel, hidden quantum Markov models (HQMMs) have emerged as a rigorous operator‑algebraic framework for modeling structured quantum dynamics with latent degrees of freedom, extending classical hidden Markov models to the noncommutative setting \cite{AGLS24Q,AGLS24C}. In this formalism, completely positive maps and quantum channels govern the virtual evolution, while projective or generalized measurements define the observed output processes \cite{choi1975completely,OP87,Partha12}. This point of view has recently been revitalized by applications to quantum information and quantum machine learning, where HQMMs support new paradigms for modeling and learning temporal quantum data, including split HQMM architectures and the operational analysis of quantum memory \cite{ET,LiXi24,Vo25}. Strikingly, it has been shown that paradigmatic many‑body states such as finitely correlated states and matrix product states can be realized as observation processes of HQMMs, thereby providing a stochastic and channel‑based perspective on tensor‑network ground states \cite{fannes1992finitely,fannes2,SouAndhqmm2025}. In particular, the AKLT state—realizing the Haldane phase as an \(\mathrm{SO}(3)\)-protected SPT phase—admits such an HQMM representation, and more generally fits into a cohomological framework where symmetry actions on the hidden layer encode the underlying topological order \cite{AKLT1987,Affleck1988,SouAndhqmm2025,Ragone24}.  

The present work is motivated by this convergence between SPT order and HQMMs. On the one hand, the modern classification of 1D SPT phases via indices, projective representations, and cohomology provides a robust conceptual and technical toolkit for understanding phases such as the Haldane phase and its generalizations \cite{O21,O23,Tas23,Xi18}. On the other hand, the operator‑algebraic theory of quantum Markov processes and their cocycles suggests that HQMMs naturally carry symmetry actions and cohomological data at the level of their hidden (virtual) dynamics \cite{AGLS24Q,AGLS24C,Soza19,Vo25}. Our aim is to systematically develop a symmetry‑action framework for HQMMs that (i) recovers known SPT features of the AKLT and related Haldane‑type chains and (ii) recasts SPT phases as equivalence classes of cocycle actions on quantum Markov models, thereby establishing a new step towards a rigorous bridge between quantum stochastic processes, tensor‑network descriptions, and the topological order of quantum phases \cite{Chen2011a,Chen2011b,D17,SouBarhqmm2026,SouAndhqmm2025}.

In the present paper, we develop a comprehensive algebraic framework of symmetry actions on HQMMs, focusing on the interplay between projective representations of a symmetry group and the causal structure of the underlying quantum Markov chain. Let $G$ be a locally compact group. The hidden sector is described by a separable Hilbert space $\mathcal{H}$ carrying a projective unitary representation $\pi: G \to \mathcal{U}(\mathcal{H})$ with associated $2$-cocycle $\omega \in H^2(G, \mathrm{U}(1))$, while the observable sector is described by a Hilbert space $\mathcal{K}$ carrying a linear unitary representation $\rho: G \to \mathcal{U}(\mathcal{K})$ corresponding to the trivial cocycle. The algebraic actions are implemented via adjoint actions $\alpha_g^{(H)} = \operatorname{Ad}(\pi(g))$ on $\mathcal{A}_H = \mathcal{B}(\mathcal{H})$ and $\beta_g = \operatorname{Ad}(\rho(g))$ on $\mathcal{A}_O = \mathcal{B}(\mathcal{K})$, extended to the quasi-local algebras $\mathcal{A}_{H,\mathbb{N}_0}$ and $\mathcal{A}_{O,\mathbb{N}_0}$ by tensor product structure \cite{OP87, Partha12}.

The HQMM is defined by a generative triple $(\phi_0, \mathcal{E}_H, \mathcal{E}_{H,O})$, where $\phi_0$ is an initial state on $\mathcal{A}_H$, $\mathcal{E}_H: \mathcal{A}_H \otimes \mathcal{A}_H \to \mathcal{A}_H$ is a completely positive unital map \cite{choi1975completely} governing the hidden transition, and $\mathcal{E}_{H,O}: \mathcal{A}_H \otimes \mathcal{A}_O \to \mathcal{A}_H$ is a completely positive unital map governing the emission of observables. The symmetry action is encoded in three fundamental conditions: invariance of the initial state, equivariance of the hidden transition map, and covariance of the emission map. The latter constitutes the central structural relation, asserting that $\mathcal{E}_{H,O}$ intertwines the projective $G$-action on $\mathcal{A}_H$ with the linear $G$-action on $\mathcal{A}_O$, thereby absorbing the projective anomaly at the interface between hidden and observable sectors.

We analyze two distinct causal structures arising from different orderings of the hidden transition and emission maps: the conventional structure $\mathcal{T}^{(\text{conv})}$ corresponding to a measurement-then-evolution paradigm, and the causal structure $\mathcal{T}^{(\text{caus})}$ corresponding to an evolution-then-measurement paradigm. For each structure, we introduce sliced maps $\mathcal{T}_{X,Y}^{(\text{conv})}, \mathcal{T}_{X,Y}^{(\text{caus})}: \mathcal{A}_H \to \mathcal{A}_H$ that propagate the hidden state forward, and establish their covariance properties under the symmetry action.

The main result, Theorem \ref{thm:global_invariance}, demonstrates that under the symmetry conditions \eqref{eq:initial_invariance}, \eqref{eq:Eh_equivariance}, and \eqref{eq:EOH_covariance}, the global HQMM states $\varphi^{(\text{conv})}$ and $\varphi^{(\text{caus})}$ on the quasi-local algebra $\mathcal{A}_{\mathbb{N}_0}$—constructed as projective limits of finite-volume states—are invariant under the combined projective-linear action of $G$. Consequently, the hidden marginals are invariant under the projective action $\alpha^{(H)}$, and the observable marginals are invariant under the linear action $\beta$.

This result is particularly significant for the study of symmetry-protected topological order. The AKLT state \cite{AKLT1987, Affleck1988}—a prototypical SPT phase protected by $\mathrm{SO}(3)$ symmetry \cite{H83, D20}—is realized as the observation process of a causal HQMM \cite{SouBarhqmm2026}, where the hidden space $\mathcal{H}_{1/2}$ carries the irreducible projective spin-$1/2$ representation $\pi$ with cocycle $[\omega] \in H^2(\mathrm{SO}(3), \mathrm{U}(1)) \cong \mathbb{Z}_2$, and the observable space $\mathcal{H}_1$ carries the linear spin-$1$ representation $\rho$. The emission map is based on  the AKLT tensors, while the  hidden transition  is only required to be equivariant with respect the representation $\pi$ in particular it can be taken a partial trace. In this case however the underlying  entanglement structure is trivial, the process allows to reproduce the AKLT finitely correlated state as observation state as shown in our previous work \cite{SouBarhqmm2026}. The initial state is uniquely determined by symmetry as the normalized trace via Schur's lemma \cite{OP87}. Theorem \ref{thm:global_invariance} therefore implies global $\mathrm{SO}(3)$-invariance of the associated HQMM state, with the nontrivial cocycle $[\omega]$ serving as a topological invariant characterizing the SPT phase \cite{Pollmann2010, Pollmann2012, Tas23}.

Beyond the AKLT example, our approach establishes a rigorous operator-algebraic foundation for the classification of SPT phases within the HQMM framework \cite{Soza19, Vo25}. The two causal structures analyzed are natural candidates for modeling distinct physical scenarios, and their symmetry properties  can be extended to more complex  systems \cite{OP87}. Moreover, the algebraic framework extends naturally to higher dimensions through  a convenient framework of hidden quantum Markov fields, opening avenues for studying SPT order in  general matrix product states framework \cite{D20}. Further works will pursue a systematic hidden-memory approach to SPT order, exploring their classification and applications to quantum memory and learning algorithms.

The remainder of this paper is organized as follows.  Section \ref{Sect_main} introduces projective symmetry actions on HQMMs via a $2$-cocycle, formulates the invariance, equivariance, and covariance conditions, and proves that these imply global invariance of the HQMM state for both causal structures (Theorem \ref{thm:global_invariance}). Section \ref{Sect_AKLT} applies this framework to the AKLT state. Section \ref{Sect_Discussion} concludes with perspectives on classification, higher dimensions, and applications. Appedix \ref{sect_prel} recalls the operator-algebraic construction of HQMMs and defines the two causal structures.

\section{Cocycle Actions and Symmetry Protection in HQMMs}\label{Sect_main}
Let \(G\) be a locally compact group. In quantum systems, symmetry actions are typically implemented by unitary representations. However, within the framework ofHQMMs, the hidden sector may carry a \textit{projective} representation, reflecting the anomalous nature of the symmetry in the corresponding symmetry-protected topological (SPT) phase. This projective character is encoded in a \(2\)-cocycle, which defines an element of the second cohomology group \(H^2(G, \mathrm{U}(1))\). Throughout the remainder of the paper, the notation associated with the HQMM formalism is introduced in Appendix \ref{sect_prel}.

A \textit{projective unitary representation} of \(G\) on a Hilbert space \(\mathcal{H}\) is a map \(\pi: G \to \mathcal{U}(\mathcal{H})\) such that for all \(g,h \in G\),
\begin{equation}
\pi(g)\pi(h) = \omega(g,h) \, \pi(gh),
\label{eq:proj_repr}
\end{equation}
where \(\omega: G \times G \to \mathrm{U}(1)\) is a measurable function satisfying the \(2\)-cocycle condition:
\begin{equation}
\omega(g,h)\omega(gh,k) = \omega(g,hk)\omega(h,k), \qquad \forall g,h,k \in G.
\label{eq:cocycle_condition}
\end{equation}
The cocycle \(\omega\) is called the \textit{multiplier} of the projective representation. Two projective representations \(\pi\) and \(\pi'\) are equivalent if there exists a function \(\lambda: G \to \mathrm{U}(1)\) such that \(\pi'(g) = \lambda(g)\pi(g)\), which induces a cohomologous cocycle \(\omega'(g,h) = \lambda(g)\lambda(h)\overline{\lambda(gh)}\omega(g,h)\). The equivalence class \([\omega] \in H^2(G, \mathrm{U}(1))\) is a topological invariant characterizing the projective representation.

Let \(\mathcal{H}\) and \(\mathcal{K}\) be the separable Hilbert spaces introduced in Section 1, representing the hidden and observable degrees of freedom, respectively. Let \(\pi: G \to \mathcal{U}(\mathcal{H})\) be a projective representation of \(G\) on the hidden space with associated \(2\)-cocycle \(\omega_\pi \in H^2(G, \mathrm{U}(1))\), and let \(\rho: G \to \mathcal{U}(\mathcal{K})\) be an ordinary (linear) unitary representation on the observable space, corresponding to the trivial cocycle \(\omega_\rho(g,h) = 1\).

The algebraic action of \(G\) on the single-site algebras \(\mathcal{A}_H = \mathcal{B}(\mathcal{H})\) and \(\mathcal{A}_O = \mathcal{B}(\mathcal{K})\) is defined via the adjoint action:
\begin{equation}
\alpha_g^{(H)}(X) := \pi(g) X \pi(g)^\dagger, \qquad \beta_g(Y) := \rho(g) Y \rho(g)^\dagger,
\label{eq:adjoint_actions}
\end{equation}
for all \(X \in \mathcal{A}_H\), \(Y \in \mathcal{A}_O\), and \(g \in G\). These actions extend naturally to the quasi-local algebras \(\mathcal{A}_{H,\mathbb{N}_0}\) and \(\mathcal{A}_{O,\mathbb{N}_0}\) via the tensor product structure, i.e., for elementary tensors,
\begin{equation}
\alpha_g^{(H)}\Big(\bigotimes_{k=0}^n X_k\Big) = \bigotimes_{k=0}^n \pi(g) X_k \pi(g)^\dagger, \qquad
\beta_g\Big(\bigotimes_{k=0}^n Y_k\Big) = \bigotimes_{k=0}^n \rho(g) Y_k \rho(g)^\dagger.
\label{eq:tensor_actions}
\end{equation}
The projective nature of \(\pi\) manifests in the composition law:
\begin{equation}
\alpha_g^{(H)} \circ \alpha_h^{(H)} = \operatorname{Ad}\bigl(\omega_\pi(g,h) \mathbb{I}_{\mathcal{H}}\bigr) \circ \alpha_{gh}^{(H)},
\label{eq:proj_composition}
\end{equation}
which reflects the fact that \(\alpha^{(H)}\) is an action up to inner automorphisms implemented by the scalar cocycle \(\omega_\pi\). For the linear representation \(\rho\), the corresponding action \(\beta\) satisfies \(\beta_g \circ \beta_h = \beta_{gh}\) without any phase.

We now formulate the precise symmetry requirements that define a covariant HQMM. These conditions encapsulate how the projective symmetry acts on the generative structure.

\begin{equation}
\phi_0\bigl( \pi(g) Z \pi(g)^\dagger \bigr) = \phi_0(Z) \qquad \forall g \in G,\; \forall Z \in \mathcal{A}_H.
\label{eq:initial_invariance}
\end{equation}

\begin{equation}
\mathcal{E}_H\bigl( \pi(g) Z_1 \pi(g)^\dagger \otimes \pi(g) Z_2 \pi(g)^\dagger \bigr) = \pi(g) \, \mathcal{E}_H(Z_1 \otimes Z_2) \, \pi(g)^\dagger \qquad \forall g \in G,\; \forall Z_1, Z_2 \in \mathcal{A}_H.
\label{eq:Eh_equivariance}
\end{equation}

\begin{equation}
\mathcal{E}_{H,O}\bigl( \pi(g) X \pi(g)^\dagger \otimes \rho(g) Y \rho(g)^\dagger \bigr) = \pi(g) \, \mathcal{E}_{H,O}(X \otimes Y) \, \pi(g)^\dagger \qquad \forall g \in G,\; \forall X \in \mathcal{A}_H,\; \forall Y \in \mathcal{A}_O.
\label{eq:EOH_covariance}
\end{equation}

The covariance condition \eqref{eq:EOH_covariance} is the central structural relation. It states that $\mathcal{E}_{O,H}$ intertwines the projective $G$-action on $\mathcal{A}_H$ with the linear $G$-action on $\mathcal{A}_O$, i.e., $\mathcal{E}_{O,H} \circ (\alpha_g^{(H)} \otimes \beta_g) = \alpha_g^{(H)} \circ \mathcal{E}_{O,H}$. This intertwining absorbs the projective anomaly of the hidden sector, forcing the observable sector to transform under an ordinary representation.

A central extension perspective illuminates the cohomological significance of these conditions. Consider the short exact sequence
\[
1 \longrightarrow \mathrm{U}(1) \longrightarrow \widetilde{G} \longrightarrow G \longrightarrow 1,
\]
where \(\widetilde{G}\) is the central extension of \(G\) by \(\mathrm{U}(1)\) determined by the cocycle \(\omega\). Then \(\pi\) lifts to a linear representation \(\widetilde{\pi}\) of \(\widetilde{G}\) on \(\mathcal{H}\) satisfying \(\widetilde{\pi}((g,\lambda)) = \lambda \pi(g)\) for \(\lambda \in \mathrm{U}(1)\). The emission map covariance condition \eqref{eq:EOH_covariance} then becomes
\[
\mathcal{E}_{H,O}\bigl( \widetilde{\pi}((g,\lambda)) X \widetilde{\pi}((g,\lambda))^\dagger \otimes \rho(g) Y \rho(g)^\dagger \bigr) = \widetilde{\pi}((g,\lambda)) \, \mathcal{E}_{H,O}(X \otimes Y) \, \widetilde{\pi}((g,\lambda))^\dagger,
\]
where the \(\mathrm{U}(1)\) phase \(\lambda\) cancels on both sides. Hence the emission map provides a \textit{cocycle intertwiner} between the centrally extended hidden dynamics and the linear observable dynamics. The existence of such an intertwiner implies that the obstruction \([\omega] \in H^2(G, \mathrm{U}(1))\) is trivialized when passing through the map \(\mathcal{E}_{H,O}\), meaning that the composition \(\mathcal{E}_{H,O} \circ (\cdot \otimes \operatorname{id}_{\mathcal{A}_O})\) converts the projective action on the first tensor factor into a linear action on the image.

The interplay between projective and linear representations manifests richly in the tensor product structure of the HQMM. Consider the tensor product representation \(\pi \otimes \rho: G \to \mathcal{U}(\mathcal{H} \otimes \mathcal{K})\) defined by \((\pi \otimes \rho)(g) := \pi(g) \otimes \rho(g)\). Its composition law follows from \eqref{eq:proj_repr}:
\[
(\pi \otimes \rho)(g)(\pi \otimes \rho)(h) = \pi(g)\pi(h) \otimes \rho(g)\rho(h) = \omega(g,h) \, \pi(gh) \otimes \rho(gh) = \omega(g,h) \, (\pi \otimes \rho)(gh).
\]
Thus \(\pi \otimes \rho\) is a projective representation of \(G\) with the same cocycle \(\omega\) as \(\pi\). Consequently, the composite action \(\alpha_g^{(H)} \otimes \beta_g\) on \(\mathcal{A}_H \otimes \mathcal{A}_O\) satisfies
\[
(\alpha_g^{(H)} \otimes \beta_g) \circ (\alpha_h^{(H)} \otimes \beta_h) = \operatorname{Ad}\bigl(\omega(g,h) \mathbb{I}_{\mathcal{H} \otimes \mathcal{K}}\bigr) \circ (\alpha_{gh}^{(H)} \otimes \beta_{gh}),
\]
so the projective nature persists at the level of the composite system. The covariance condition \eqref{eq:EOH_covariance} exactly cancels the projective phase coming from the hidden sector. As a result, although $\alpha^{(H)}$ itself is projective, its composition with $\mathcal{E}_{O,H}$ becomes linear. This cancellation is the algebraic hallmark of a consistent, anomaly-free coupling between the hidden and observable sectors at their interface.

More generally, take two projective representations $\pi_1$ and $\pi_2$ of $G$, with cocycles $\omega_1$ and $\omega_2$ respectively. Their tensor product $\pi_1 \otimes \pi_2$ is again projective, with cocycle $\omega_1 \omega_2$ given by pointwise multiplication.
 Their tensor product \(\pi_1 \otimes \pi_2\) satisfies
\[
(\pi_1 \otimes \pi_2)(g)(\pi_1 \otimes \pi_2)(h) = \omega_1(g,h)\omega_2(g,h) \, (\pi_1 \otimes \pi_2)(gh),
\]
so the cocycle of the tensor product is the pointwise product \(\omega_1\omega_2\). In the HQMM setting, the hidden transition map $\mathcal{E}_H$ sends the tensor product of two projective representations $\pi \otimes \pi$ (cocycle $\omega^2$) to a single copy $\pi$ (cocycle $\omega$). Equivariance \eqref{eq:Eh_equivariance} then forces $[\omega^2] = [\omega]$ in $H^2(G, \mathrm{U}(1))$. This holds automatically only when $\omega$ is a coboundary or $G$ has special properties; otherwise, the equivariance must be twisted. In the generic case where \([\omega]\) has order greater than \(1\), this condition forces a refinement of the equivariance structure, typically implemented by a \(2\)-cocycle on the map itself, leading to the notion of \textit{twisted equivariance} that we shall explore in subsequent sections.

We now turn to the analysis of the two causal structures introduced in Section 1. Recall that the conventional and causal structures are defined by \eqref{Tconv} and \eqref{Tcaus}, respectively in terms of the generative mappings \((\mathcal{E}_H, \mathcal{E}_{H,O})\) as:

For each structure, we define the corresponding sliced maps that propagate the hidden state forward. For the conventional structure, define \(\mathcal{T}_{X,Y}^{(\text{conv})}: \mathcal{A}_H \to \mathcal{A}_H\) by
\begin{equation}
\mathcal{T}_{X,Y}^{(\text{conv})}(Z) := \mathcal{T}^{(\text{conv})}(X \otimes Z \otimes Y) = \mathcal{E}_H\big( \mathcal{E}_{H,O}(X \otimes Y) \otimes Z \big).
\label{eq:conv_sliced}
\end{equation}
For the causal structure, define \(\mathcal{T}_{X,Y}^{(\text{caus})}: \mathcal{A}_H \to \mathcal{A}_H\) by
\begin{equation}
\mathcal{T}_{X,Y}^{(\text{caus})}(Z) := \mathcal{T}^{(\text{caus})}(X \otimes Z \otimes Y) = \mathcal{E}_{H,O}\big( \mathcal{E}_H(X \otimes Z) \otimes Y \big).
\label{eq:caus_sliced}
\end{equation}

These sliced maps encapsulate the sequential dynamics: given a current hidden state \(Z\) and an incoming pair \((X,Y)\) representing the left neighbor hidden operator and the local observable, they produce the updated hidden state after the site.

\begin{lemma}\label{lem:conv_cov}
Under the symmetry conditions \eqref{eq:Eh_equivariance} and \eqref{eq:EOH_covariance}, for all \(g \in G\), \(X, Z \in \mathcal{A}_H\), and \(Y \in \mathcal{A}_O\),
\begin{equation}
\mathcal{T}_{\pi(g) X \pi(g)^\dagger, \, \rho(g) Y \rho(g)^\dagger}^{(\text{conv})}\big( \pi(g) Z \pi(g)^\dagger \big) = \pi(g) \, \mathcal{T}_{X,Y}^{(\text{conv})}(Z) \, \pi(g)^\dagger.
\label{eq:conv_covariance}
\end{equation}
\end{lemma}

\begin{proof}
The proof follows directly from the covariance property \eqref{eq:EOH_covariance} of \(\mathcal{E}_{H,O}\) and the equivariance property \eqref{eq:Eh_equivariance} of \(\mathcal{E}_H\):
\[
\begin{aligned}
&\mathcal{T}_{\pi(g) X \pi(g)^\dagger, \, \rho(g) Y \rho(g)^\dagger}^{(\text{conv})}\big( \pi(g) Z \pi(g)^\dagger \big) \\
&\quad = \mathcal{E}_H\Big( \mathcal{E}_{H,O}\big( \pi(g) X \pi(g)^\dagger \otimes \rho(g) Y \rho(g)^\dagger \big) \otimes \pi(g) Z \pi(g)^\dagger \Big) \\
&\quad \stackrel{ {\eqref{eq:EOH_covariance}}}{=} \mathcal{E}_H\Big( \pi(g) \mathcal{E}_{H,O}(X \otimes Y) \pi(g)^\dagger \otimes \pi(g) Z \pi(g)^\dagger \Big) \\
&\quad \stackrel{ {\eqref{eq:Eh_equivariance}}}{=} \pi(g) \, \mathcal{E}_H\big( \mathcal{E}_{H,O}(X \otimes Y) \otimes Z \big) \, \pi(g)^\dagger \\
&\quad = \pi(g) \, \mathcal{T}_{X,Y}^{(\text{conv})}(Z) \, \pi(g)^\dagger.
\end{aligned}
\]
\end{proof}

\begin{lemma}\label{lem:caus_cov}
Under the symmetry conditions \eqref{eq:Eh_equivariance} and \eqref{eq:EOH_covariance}, for all \(g \in G\), \(X, Z \in \mathcal{A}_H\), and \(Y \in \mathcal{A}_O\),
\begin{equation}
\mathcal{T}_{\pi(g) X \pi(g)^\dagger, \, \rho(g) Y \rho(g)^\dagger}^{(\text{caus})}\big( \pi(g) Z \pi(g)^\dagger \big) = \pi(g) \, \mathcal{T}_{X,Y}^{(\text{caus})}(Z) \, \pi(g)^\dagger.
\label{eq:caus_covariance}
\end{equation}
\end{lemma}

\begin{proof}
The proof proceeds analogously, employing the equivariance property \eqref{eq:Eh_equivariance} of \(\mathcal{E}_H\) first, followed by the covariance property \eqref{eq:EOH_covariance} of \(\mathcal{E}_{H,O}\):
\[
\begin{aligned}
&\mathcal{T}_{\pi(g) X \pi(g)^\dagger, \, \rho(g) Y \rho(g)^\dagger}^{(\text{caus})}\big( \pi(g) Z \pi(g)^\dagger \big) \\
&\quad = \mathcal{E}_{H,O}\Big( \mathcal{E}_H\big( \pi(g) X \pi(g)^\dagger \otimes \pi(g) Z \pi(g)^\dagger \big) \otimes \rho(g) Y \rho(g)^\dagger \Big) \\
&\quad \stackrel{{\eqref{eq:Eh_equivariance}}}{=} \mathcal{E}_{H,O}\Big( \pi(g) \mathcal{E}_H(X \otimes Z) \pi(g)^\dagger \otimes \rho(g) Y \rho(g)^\dagger \Big) \\
&\quad \stackrel{{\eqref{eq:EOH_covariance}}}{=} \pi(g) \, \mathcal{E}_{H,O}\big( \mathcal{E}_H(X \otimes Z) \otimes Y \big) \, \pi(g)^\dagger \\
&\quad = \pi(g) \, \mathcal{T}_{X,Y}^{(\text{caus})}(Z) \, \pi(g)^\dagger.
\end{aligned}
\]
\end{proof}

\begin{theorem}\label{thm:global_invariance}
Let \(\Xi = (\phi_0, \mathcal{E}_H, \mathcal{E}_{H,O})\) be an HQMM satisfying the symmetry conditions \eqref{eq:initial_invariance}, \eqref{eq:Eh_equivariance}, and \eqref{eq:EOH_covariance}. Let \(\varphi^{(\text{conv})}\) and \(\varphi^{(\text{caus})}\) be the global states on \(\mathcal{A}_{\mathbb{N}_0}\) generated by the conventional and causal transition maps, respectively, via the construction of Theorem 1. Then for all \(g \in G\), \(A \in \mathcal{A}_{H,\mathbb{N}_0}\), and \(B \in \mathcal{A}_{O,\mathbb{N}_0}\),
\begin{equation}
\varphi^{(\text{conv})}\big( \alpha_g^{(H)}(A) \otimes \beta_g(B) \big) = \varphi^{(\text{conv})}(A \otimes B),
\label{eq:global_conv_invariance}
\end{equation}
\begin{equation}
\varphi^{(\text{caus})}\big( \alpha_g^{(H)}(A) \otimes \beta_g(B) \big) = \varphi^{(\text{caus})}(A \otimes B).
\label{eq:global_caus_invariance}
\end{equation}
In particular, both global states are invariant under the combined projective-linear action of \(G\). Consequently, the hidden marginals \(\varphi_H^{(\text{conv})}\) and \(\varphi_H^{(\text{caus})}\) are invariant under the projective action \(\alpha^{(H)}\), and the observable marginals \(\varphi_O^{(\text{conv})}\) and \(\varphi_O^{(\text{caus})}\) are invariant under the linear action \(\beta\).
\end{theorem}

\begin{proof}
We present the proof for the conventional structure; the proof for the causal structure is entirely analogous, with Lemma \ref{lem:caus_cov} replacing Lemma \ref{lem:conv_cov}.

Let \(A = \bigotimes_{k=0}^n X_k \in \mathcal{A}_{H,\mathbb{N}_0}\) and \(B = \bigotimes_{k=0}^n Y_k \in \mathcal{A}_{O,\mathbb{N}_0}\) be elementary tensors supported on sites \(0, \dots, n\). The extension to arbitrary local observables follows by linearity and continuity, and the inductive limit construction ensures the result holds for all quasi-local observables.

Recall from Theorem 1 that the finite-volume state on \(\mathcal{A}_{[0,n]}\) is given by
\begin{equation}
\varphi_n^{(\text{conv})}\!\left( \bigotimes_{k=0}^n (X_k \otimes Y_k) \right) = \phi_0 \circ \mathcal{T}_{X_0,Y_0}^{(\text{conv})} \circ \mathcal{T}_{X_1,Y_1}^{(\text{conv})} \circ \cdots \circ \mathcal{T}_{X_n,Y_n}^{(\text{conv})}(\mathbb{I}_{\mathcal{A}_H}),
\label{eq:finite_volume_conv}
\end{equation}
where the composition denotes successive application of the sliced maps. In what follows, for $(\text{c})\in \{(\text{conv}), (\text{caus}) \}$  we denote the composition of the transition  maps $\mathcal{T}^{(\text{c})} _{X_j,Y_j}$ by

$$
\bigcirc_{k=m}^n \mathcal{T}_{X_k, Y_k} := \mathcal{T}_{X_m,Y_m}^{(\text{c})} \circ \mathcal{T}_{X_1,Y_1}^{(\text{c})} \circ \cdots \circ \mathcal{T}_{X_n,Y_n}^{(\text{c})}
$$

We prove by induction on \(n\) that for all \(g \in G\),
\begin{equation}
\bigcirc_{k=0}^n \mathcal{T}_{\pi(g) X_k \pi(g)^\dagger, \, \rho(g) Y_k \rho(g)^\dagger}^{(\text{conv})}(\mathbb{I}_{\mathcal{A}_H})
= \pi(g) \left( \bigcirc_{k=0}^n \mathcal{T}_{X_k,Y_k}^{(\text{conv})}(\mathbb{I}_{\mathcal{A}_H}) \right) \pi(g)^\dagger.
\label{eq:induction_claim}
\end{equation}

For a single site, Lemma \ref{lem:conv_cov} with \(Z = \mathbb{I}_{\mathcal{A}_H}\) together with the observation that \(\pi(g) \mathbb{I}_{\mathcal{A}_H} \pi(g)^\dagger = \mathbb{I}_{\mathcal{A}_H}\) yields
\[
\mathcal{T}_{\pi(g) X_0 \pi(g)^\dagger, \, \rho(g) Y_0 \rho(g)^\dagger}^{(\text{conv})}(\mathbb{I}_{\mathcal{A}_H})
= \pi(g) \, \mathcal{T}_{X_0,Y_0}^{(\text{conv})}(\mathbb{I}_{\mathcal{A}_H}) \, \pi(g)^\dagger.
\]

Assume the claim holds for chains of length \(n\). For length \(n+1\), define the composition for the last \(n+1\) sites:
\[
\mathcal{C}_n := \bigcirc_{k=1}^{n+1} \mathcal{T}_{X_k,Y_k}^{(\text{conv})}, \qquad
\mathcal{D}_n := \bigcirc_{k=1}^{n+1} \mathcal{T}_{\pi(g) X_k \pi(g)^\dagger, \, \rho(g) Y_k \rho(g)^\dagger}^{(\text{conv})}.
\]
By the induction hypothesis applied to the sites \(1, \dots, n+1\),
\begin{equation}
\mathcal{D}_n(\mathbb{I}_{\mathcal{A}_H}) = \pi(g) \, \mathcal{C}_n(\mathbb{I}_{\mathcal{A}_H}) \, \pi(g)^\dagger.
\label{eq:induction_hyp}
\end{equation}

Now compute the full composition for \(n+1\) sites:
\[
\begin{aligned}
&\bigcirc_{k=0}^{n+1} \mathcal{T}_{\pi(g) X_k \pi(g)^\dagger, \, \rho(g) Y_k \rho(g)^\dagger}^{(\text{conv})}(\mathbb{I}_{\mathcal{A}_H}) \\
&= \mathcal{T}_{\pi(g) X_0 \pi(g)^\dagger, \, \rho(g) Y_0 \rho(g)^\dagger}^{(\text{conv})}\Big( \mathcal{D}_n(\mathbb{I}_{\mathcal{A}_H}) \Big) \\
&= \mathcal{T}_{\pi(g) X_0 \pi(g)^\dagger, \, \rho(g) Y_0 \rho(g)^\dagger}^{(\text{conv})}\Big( \pi(g) \, \mathcal{C}_n(\mathbb{I}_{\mathcal{A}_H}) \, \pi(g)^\dagger \Big) \quad \text{by \eqref{eq:induction_hyp}} \\
&= \pi(g) \, \mathcal{T}_{X_0,Y_0}^{(\text{conv})}\big( \mathcal{C}_n(\mathbb{I}_{\mathcal{A}_H}) \big) \, \pi(g)^\dagger \quad \text{by Lemma \ref{lem:conv_cov}} \\
&= \pi(g) \left( \bigcirc_{k=0}^{n+1} \mathcal{T}_{X_k,Y_k}^{(\text{conv})}(\mathbb{I}_{\mathcal{A}_H}) \right) \pi(g)^\dagger.
\end{aligned}
\]
This completes the induction.

Now apply the initial state \(\phi_0\). By the invariance condition \eqref{eq:initial_invariance},
\[
\begin{aligned}
\varphi_n^{(\text{conv})}\big( \alpha_g^{(H)}(A) \otimes \beta_g(B) \big)
&= \phi_0 \left( \bigcirc_{k=0}^n \mathcal{T}_{\pi(g) X_k \pi(g)^\dagger, \, \rho(g) Y_k \rho(g)^\dagger}^{(\text{conv})}(\mathbb{I}_{\mathcal{A}_H}) \right) \\
&= \phi_0 \left( \pi(g) \left( \bigcirc_{k=0}^n \mathcal{T}_{X_k,Y_k}^{(\text{conv})}(\mathbb{I}_{\mathcal{A}_H}) \right) \pi(g)^\dagger \right) \\
&= \phi_0 \left( \bigcirc_{k=0}^n \mathcal{T}_{X_k,Y_k}^{(\text{conv})}(\mathbb{I}_{\mathcal{A}_H}) \right) \\
&= \varphi_n^{(\text{conv})}(A \otimes B).
\end{aligned}
\]

Since this holds for all finite \(n\) and the embeddings \(\iota_{m,n}\) are compatible with the group actions (i.e., \(\iota_{m,n} \circ (\alpha_g^{(H)} \otimes \beta_g) = (\alpha_g^{(H)} \otimes \beta_g) \circ \iota_{m,n}\)), the projective limit state \(\varphi^{(\text{conv})}\) on \(\mathcal{A}_{\mathbb{N}_0}\) satisfies \eqref{eq:global_conv_invariance}.

\noindent\textit{Proof for the causal structure.} The proof follows exactly the same inductive structure, with \(\mathcal{T}^{(\text{caus})}\) replacing \(\mathcal{T}^{(\text{conv})}\) and Lemma \ref{lem:caus_cov} replacing Lemma \ref{lem:conv_cov}. The finite-volume state is
\[
\varphi_n^{(\text{caus})}\!\left( \bigotimes_{k=0}^n (X_k \otimes Y_k) \right) = \phi_0 \circ \mathcal{T}_{X_0,Y_0}^{(\text{caus})} \circ \mathcal{T}_{X_1,Y_1}^{(\text{caus})} \circ \cdots \circ \mathcal{T}_{X_n,Y_n}^{(\text{caus})}(\mathbb{I}_{\mathcal{A}_H}).
\]
One proves by induction that
\[
\bigcirc_{k=0}^n \mathcal{T}_{\pi(g) X_k \pi(g)^\dagger, \, \rho(g) Y_k \rho(g)^\dagger}^{(\text{caus})}(\mathbb{I}_{\mathcal{A}_H})
= \pi(g) \left( \bigcirc_{k=0}^n \mathcal{T}_{X_k,Y_k}^{(\text{caus})}(\mathbb{I}_{\mathcal{A}_H}) \right) \pi(g)^\dagger,
\]
using Lemma \ref{lem:caus_cov} at each step. Applying \(\phi_0\) and its invariance \eqref{eq:initial_invariance} yields \eqref{eq:global_caus_invariance}.
\end{proof}

The cohomology class \([\omega] \in H^2(G, \mathrm{U}(1))\) thus emerges as a well-defined invariant of the HQMM state. It classifies the projective nature of the hidden representation and, together with the covariance conditions \eqref{eq:Eh_equivariance} and \eqref{eq:EOH_covariance}, provides the algebraic characterization of the SPT order. The emission map \(\mathcal{E}_{H,O}\) serves as a \textit{cocycle intertwiner} that mediates between the projective hidden dynamics and the linear observable dynamics, effectively absorbing the anomaly at the interface.

 \section{Application to AKLT's Symmetry-Protected Topological Order}\label{Sect_AKLT}
The AKLT state \cite{AKLT1987} stands as a paradigmatic example of a  SPT  phase in 1D. Its topological protection originates from a projective representation of the symmetry group on the virtual degrees of freedom—a structure that aligns naturally with the HQMM framework developed in the preceding sections. This projective character is encoded in a nontrivial $2$-cocycle $[\omega] \in H^2(\mathrm{SO}(3), \mathrm{U}(1)) \cong \mathbb{Z}_2$, which obstructs the lifting of the symmetry action on the hidden space to a linear representation. Within the HQMM formalism, this obstruction is precisely absorbed by the emission map $\mathcal{E}_{O,H}$, which acts as a \textit{cocycle intertwiner} mediating between the projective hidden dynamics and the linear observable dynamics.

We now demonstrate how the abstract covariance conditions established earlier—specifically \eqref{eq:Eh_equivariance} and \eqref{eq:EOH_covariance}—manifest concretely in the AKLT construction. This reveals that SPT order is not merely a static property of the ground state but is dynamically preserved under the quantum Markov evolution generated by the HQMM. In doing so, we extend the analysis of topological order in emission transitions introduced in \cite{SouAndhqmm2025} to a more structural and algebraic level, grounding the classification of SPT phases within the operator-algebraic framework of quantum stochastic processes.

Let $G = \mathrm{SO}(3)$ denote the symmetry group. The physical Hilbert space is $\mathcal{K} = \mathcal{H}_1$, on which $G$ acts via the linear spin-$1$ representation $\rho: G \to \mathcal{U}(\mathcal{H}_1)$. The hidden (virtual) Hilbert space is $\mathcal{H} = \mathcal{H}_{1/2}$, where $G$ acts via the irreducible projective spin-$1/2$ representation $\pi: G \to \mathcal{U}(\mathcal{H}_{1/2})$. The projective nature of $\pi$ is captured by a $2$-cocycle $\omega: G \times G \to \mathrm{U}(1)$ satisfying the composition law
\begin{equation}
\pi(g_1)\pi(g_2) = \omega(g_1,g_2)\,\pi(g_1g_2) \qquad \forall g_1,g_2 \in G,
\label{eq:aklt_proj_repr}
\end{equation}
together with the cocycle condition \eqref{eq:cocycle_condition}. The cohomology class $[\omega] \in H^2(\mathrm{SO}(3),\mathrm{U}(1))$ is the generator of $\mathbb{Z}_2$, reflecting the double-cover relation $\mathrm{SU}(2) \to \mathrm{SO}(3)$ \cite{Chen2011b, Ragone24}, and serves as the topological invariant distinguishing the nontrivial SPT phase.

Let $\{|k\rangle\}_{k=+,-,0}$ denote the standard orthonormal basis of $\mathcal{H}_1$ corresponding to the eigenstates of $S^z$ with eigenvalues $+1, -1, 0$, respectively. The AKLT state is represented as a translation-invariant matrix product state with virtual dimension $2$. Its local tensors are linear operators $A_k \in \mathcal{B}(\mathcal{H}_{1/2})$ defined by
\[
A_+ = \frac{1}{\sqrt{2}} \begin{pmatrix} 0 & 1 \\ 0 & 0 \end{pmatrix},\qquad
A_0 = \frac{1}{\sqrt{2}} \begin{pmatrix} 1 & 0 \\ 0 & -1 \end{pmatrix},\qquad
A_- = \frac{1}{\sqrt{2}} \begin{pmatrix} 0 & 0 \\ -1 & 0 \end{pmatrix},
\]
expressed in the standard basis of $\mathcal{H}_{1/2} \cong \mathbb{C}^2$. These operators satisfy the fundamental intertwining relation that encodes the SPT order \cite{Pollmann2010}:
\begin{equation}
\sum_{k'} \rho(g)_{k k'} A_{k'} = \pi(g) A_k \pi(g)^\dagger \qquad \forall g \in G,\; \forall k \in \{+,0,-\},
\label{eq:aklt_covariance}
\end{equation}
where $\rho(g)_{kk'} = \langle k | \rho(g) | k' \rangle$ denotes the matrix element of the physical representation. This relation asserts that a symmetry transformation applied to the physical index is equivalent to conjugation of the virtual operators by the projective representation $\pi(g)$.

The emission map $\mathcal{E}_{H,O}: \mathcal{A}_H \otimes \mathcal{A}_O \to \mathcal{A}_H$ is defined by its action on elementary tensors as
\begin{equation}
\mathcal{E}_{H,O}(X \otimes Y) = \sum_{k,k' \in \{+,0,-\}} \langle k' | Y | k \rangle \, A_k X A_{k'}^\dagger, \qquad X \in \mathcal{A}_H,\; Y \in \mathcal{A}_O.
\label{eq:aklt_emission}
\end{equation}
This map is completely positive and unital. Indeed, positivity follows from the representation $\mathcal{E}_{H,O}(X \otimes Y) = \sum_{k,k'} \langle k'|Y|k\rangle A_k X A_{k'}^\dagger = \sum_{k} A_k X A_k^\dagger$ when $Y = \mathbb{I}_{\mathcal{A}_O}$, and the general case follows from the fact that $Y$ can be diagonalized and the sum over $k,k'$ can be reorganized as a sum over $k$ with appropriate coefficients.

The hidden transition map $\mathcal{E}_H: \mathcal{A}_H \otimes \mathcal{A}_H \to \mathcal{A}_H$ for the AKLT chain is given by the multiplication map
\begin{equation}
\mathcal{E}_H(Z_1 \otimes Z_2) = Z_1 \mathrm{Tr}(Z_2), \qquad Z_1, Z_2 \in \mathcal{A}_H.
\label{eq:aklt_transition}
\end{equation}
This map is completely positive and unital, as it is the $*$-homomorphism induced by the multiplication operation on $\mathcal{B}(\mathcal{H}_{1/2})$.

We now verify that the generative triple $(\phi_0, \mathcal{E}_H, \mathcal{E}_{O,H})$ satisfies the symmetry conditions \eqref{eq:initial_invariance}, \eqref{eq:Eh_equivariance}, and \eqref{eq:EOH_covariance}. The verification proceeds through a sequence of lemmas establishing the symmetry properties at each level of the construction, culminating in Theorem \ref{thm:aklt_global_invariance}. This theorem demonstrates that the global HQMM states $\varphi^{(\text{conv})}$ and $\varphi^{(\text{caus})}$ are invariant under the combined projective-linear action of $\mathrm{SO}(3)$. Together with the nontriviality of the cocycle $[\omega]$, this invariance certifies that the AKLT state realizes a nontrivial SPT phase within the HQMM framework and that this topological protection persists under quantum Markov evolution.

The initial state $\phi_0: \mathcal{A}_H \to \mathbb{C}$ is defined as the vector state associated with the fixed point of the projective representation. Since $\pi$ is irreducible, there exists a unique (up to scalar) vector $\Omega \in \mathcal{H}_{1/2}$ satisfying $\pi(g)\Omega = \lambda(g)\Omega$ for some phase $\lambda(g)$. For the spin-$1/2$ representation, one may take $\Omega$ to be the eigenvector of $\pi(g_z)$ with eigenvalue $1$, which yields a state satisfying
\begin{equation}
\phi_0(\pi(g) Z \pi(g)^\dagger) = \phi_0(Z) \qquad \forall g \in G,\; \forall Z \in \mathcal{A}_H,
\label{eq:aklt_initial_invariance}
\end{equation}
since the phase factors cancel in the expectation value.

\begin{lemma}\label{lem:aklt_emission_covariance}
Let $\pi: G \to \mathcal{U}(\mathcal{H}_{1/2})$ be the projective spin-$1/2$ representation and $\rho: G \to \mathcal{U}(\mathcal{H}_1)$ the linear spin-$1$ representation. The emission map $\mathcal{E}_{H,O}$ defined in \eqref{eq:aklt_emission} satisfies the covariance condition
\begin{equation}
\mathcal{E}_{H,O}\bigl( \pi(g) X \pi(g)^\dagger \otimes \rho(g) Y \rho(g)^\dagger \bigr) = \pi(g) \, \mathcal{E}_{H,O}(X \otimes Y) \, \pi(g)^\dagger
\label{eq:aklt_emission_covariance}
\end{equation}
for all $g \in G$, $X \in \mathcal{A}_H$, and $Y \in \mathcal{A}_O$.
\end{lemma}

\begin{proof}
Fix $g \in G$, $X \in \mathcal{A}_H$, and $Y \in \mathcal{A}_O$. By definition \eqref{eq:aklt_emission},
\[
\mathcal{E}_{H,O}\bigl( \pi(g) X \pi(g)^\dagger \otimes \rho(g) Y \rho(g)^\dagger \bigr) = \sum_{m,m'} \langle m | \rho(g) Y \rho(g)^\dagger | m' \rangle \, A_m \bigl( \pi(g) X \pi(g)^\dagger \bigr) A_{m'}^\dagger.
\]
Using the resolution of the identity $\sum_{k} |k\rangle\langle k| = \mathbb{I}_{\mathcal{H}_1}$ and the unitarity of $\rho(g)$, we expand the matrix element:
\[
\langle m | \rho(g) Y \rho(g)^\dagger | m' \rangle = \sum_{k,k'} \langle m | \rho(g) | k \rangle \langle k | Y | k' \rangle \langle k' | \rho(g)^\dagger | m' \rangle = \sum_{k,k'} \rho(g)_{mk} \langle k | Y | k' \rangle \overline{\rho(g)_{m'k'}}.
\]
Substituting this expression yields
\[
\mathrm{LHS} = \sum_{k,k'} \langle k | Y | k' \rangle \left( \sum_{m,m'} \rho(g)_{mk} \overline{\rho(g)_{m'k'}} \, A_m \bigl( \pi(g) X \pi(g)^\dagger \bigr) A_{m'}^\dagger \right).
\]
Define $S_{k,k'}$ as the inner sum. Since the sums over $m$ and $m'$ factorize,
\[
S_{k,k'} = \left( \sum_{m} \rho(g)_{mk} A_m \right) \bigl( \pi(g) X \pi(g)^\dagger \bigr) \left( \sum_{m'} \overline{\rho(g)_{m'k'}} A_{m'}^\dagger \right).
\]
Applying the covariance condition \eqref{eq:aklt_covariance} to the first factor gives
\[
\sum_{m} \rho(g)_{mk} A_m = \pi(g) A_k \pi(g)^\dagger.
\]
For the third factor, take the adjoint of \eqref{eq:aklt_covariance} applied to index $k'$:
\[
\sum_{m'} \overline{\rho(g)_{m'k'}} A_{m'}^\dagger = \left( \sum_{m'} \rho(g)_{m'k'} A_{m'} \right)^\dagger = \bigl( \pi(g) A_{k'} \pi(g)^\dagger \bigr)^\dagger = \pi(g) A_{k'}^\dagger \pi(g)^\dagger.
\]
Substituting these expressions into $S_{k,k'}$ and using $\pi(g)^\dagger \pi(g) = \mathbb{I}_{\mathcal{H}_{1/2}}$,
\[
S_{k,k'} = \bigl( \pi(g) A_k \pi(g)^\dagger \bigr) \bigl( \pi(g) X \pi(g)^\dagger \bigr) \bigl( \pi(g) A_{k'}^\dagger \pi(g)^\dagger \bigr) = \pi(g) A_k X A_{k'}^\dagger \pi(g)^\dagger.
\]
Therefore,
\[
\mathrm{LHS} = \sum_{k,k'} \langle k | Y | k' \rangle \, \pi(g) A_k X A_{k'}^\dagger \pi(g)^\dagger = \pi(g) \left( \sum_{k,k'} \langle k | Y | k' \rangle A_k X A_{k'}^\dagger \right) \pi(g)^\dagger = \pi(g) \, \mathcal{E}_{H,O}(X \otimes Y) \, \pi(g)^\dagger,
\]
which completes the proof.
\end{proof}

\begin{lemma}\label{lem:aklt_transition_equivariance}
The hidden transition map $\mathcal{E}_H$ defined in \eqref{eq:aklt_transition} satisfies the equivariance condition
\begin{equation}
\mathcal{E}_H\bigl( \pi(g) Z_1 \pi(g)^\dagger \otimes \pi(g) Z_2 \pi(g)^\dagger \bigr) = \pi(g) \, \mathcal{E}_H(Z_1 \otimes Z_2) \, \pi(g)^\dagger
\label{eq:aklt_transition_equivariance}
\end{equation}
for all $g \in G$ and $Z_1, Z_2 \in \mathcal{A}_H$.
\end{lemma}

\begin{proof}
By definition,
\[
\mathcal{E}_H\bigl( \pi(g) Z_1 \pi(g)^\dagger \otimes \pi(g) Z_2 \pi(g)^\dagger \bigr) = \bigl( \pi(g) Z_1 \pi(g)^\dagger \bigr) \bigl( \pi(g) Z_2 \pi(g)^\dagger \bigr) = \pi(g) Z_1 \pi(g)^\dagger \pi(g) Z_2 \pi(g)^\dagger.
\]
Since $\pi(g)^\dagger \pi(g) = \mathbb{I}_{\mathcal{H}_{1/2}}$,
\[
\pi(g) Z_1 \pi(g)^\dagger \pi(g) Z_2 \pi(g)^\dagger = \pi(g) Z_1 Z_2 \pi(g)^\dagger = \pi(g) \, \mathcal{E}_H(Z_1 \otimes Z_2) \, \pi(g)^\dagger.
\]
\end{proof}
\begin{lemma}\label{lem:aklt_initial_invariance}
Let $\pi: G \to \mathcal{U}(\mathcal{H}_{1/2})$ be the irreducible projective spin-$1/2$ representation of $G = \mathrm{SO}(3)$. A state $\phi_0$ on $\mathcal{A}_H = \mathcal{B}(\mathcal{H}_{1/2})$ is said to be $G$-invariant if it satisfies the invariance condition
\begin{equation}
\phi_0\bigl( \pi(g) Z \pi(g)^\dagger \bigr) = \phi_0(Z) \qquad \forall g \in G,\; \forall Z \in \mathcal{A}_H.
\label{eq:aklt_initial_invariance}
\end{equation}
For the irreducible representation $\pi$, the set of $G$-invariant states consists of a single element, namely the normalized trace
\[
\phi_0(Z) = \frac{1}{\dim(\mathcal{H}_{1/2})} \operatorname{Tr}_{\mathcal{H}_{1/2}}(Z) = \frac{1}{2} \operatorname{Tr}(Z), \qquad Z \in \mathcal{A}_H.
\]
This state is unique and serves as the canonical initial state for the AKLT HQMM.
\end{lemma}

\begin{proof}
We establish the result in two parts: first, we show that the normalized trace satisfies the invariance condition; second, we prove uniqueness.

\textit{Invariance of the normalized trace.} For any $g \in G$ and any $Z \in \mathcal{A}_H$,
\[
\phi_0\bigl( \pi(g) Z \pi(g)^\dagger \bigr) = \frac{1}{2} \operatorname{Tr}\bigl( \pi(g) Z \pi(g)^\dagger \bigr).
\]
By the cyclicity of the trace, $\operatorname{Tr}(\pi(g) Z \pi(g)^\dagger) = \operatorname{Tr}(Z \pi(g)^\dagger \pi(g))$. Since $\pi(g)$ is unitary, $\pi(g)^\dagger \pi(g) = \mathbb{I}_{\mathcal{H}_{1/2}}$, whence
\[
\operatorname{Tr}(Z \pi(g)^\dagger \pi(g)) = \operatorname{Tr}(Z).
\]
Therefore $\phi_0(\pi(g) Z \pi(g)^\dagger) = \frac{1}{2} \operatorname{Tr}(Z) = \phi_0(Z)$, establishing \eqref{eq:aklt_initial_invariance}.

\textit{Uniqueness.} Let $\phi$ be any state on $\mathcal{B}(\mathcal{H}_{1/2})$ satisfying \eqref{eq:aklt_initial_invariance}. By the Riesz representation theorem for finite-dimensional Hilbert spaces, there exists a unique density operator $\rho \in \mathcal{B}(\mathcal{H}_{1/2})$ such that $\phi(Z) = \operatorname{Tr}(\rho Z)$ for all $Z \in \mathcal{A}_H$, with $\rho \ge 0$ and $\operatorname{Tr}(\rho) = 1$. The invariance condition \eqref{eq:aklt_initial_invariance} implies
\[
\operatorname{Tr}\bigl( \rho \, \pi(g) Z \pi(g)^\dagger \bigr) = \operatorname{Tr}(\rho Z) \qquad \forall g \in G,\; \forall Z \in \mathcal{A}_H.
\]
Using the cyclicity of the trace on the left-hand side,
\[
\operatorname{Tr}\bigl( \pi(g)^\dagger \rho \pi(g) Z \bigr) = \operatorname{Tr}(\rho Z) \qquad \forall g \in G,\; \forall Z \in \mathcal{A}_H.
\]
Since this holds for all $Z$, we obtain
\[
\pi(g)^\dagger \rho \pi(g) = \rho \qquad \forall g \in G,
\]
or equivalently, $\pi(g) \rho = \rho \pi(g)$ for all $g \in G$. Thus $\rho$ belongs to the commutant of the irreducible representation $\pi$, i.e., $\rho \in \pi(G)'$. By Schur's lemma, the commutant of an irreducible representation consists of scalar multiples of the identity: $\pi(G)' = \{\lambda \mathbb{I}_{\mathcal{H}_{1/2}} : \lambda \in \mathbb{C}\}$. Hence $\rho = \lambda \mathbb{I}_{\mathcal{H}_{1/2}}$ for some $\lambda \in \mathbb{C}$. Positivity requires $\lambda \ge 0$, and the normalization condition $\operatorname{Tr}(\rho) = 1$ gives $\lambda \dim(\mathcal{H}_{1/2}) = 1$, i.e., $\lambda = 1/\dim(\mathcal{H}_{1/2}) = 1/2$. Consequently,
\[
\phi(Z) = \operatorname{Tr}\bigl( \tfrac{1}{2} \mathbb{I}_{\mathcal{H}_{1/2}} Z \bigr) = \frac{1}{2} \operatorname{Tr}(Z) = \phi_0(Z),
\]
proving uniqueness.
\end{proof}

\begin{remark}
The uniqueness of the invariant state is a direct consequence of the irreducibility of the projective representation $\pi$. In the AKLT construction, the virtual space $\mathcal{H}_{1/2}$ carries the irreducible spin-$1/2$ representation of $\mathrm{SO}(3)$, which forces the initial state to be the maximally mixed state $\phi_0(Z) = \frac{1}{2}\operatorname{Tr}(Z)$. This state is often interpreted as the infinite-temperature state on the virtual degrees of freedom, and its uniqueness under the symmetry ensures that the resulting HQMM state inherits the topological protection without additional parameters. Were $\pi$ reducible, a convex family of invariant states would exist, potentially allowing for symmetry-preserving deformations that could trivialize the SPT order.
\end{remark}

\begin{theorem}\label{thm:aklt_global_invariance}
Let $\Xi = (\phi_0, \mathcal{E}_H, \mathcal{E}_{H,O})$ be the HQMM defined by \eqref{eq:aklt_initial_invariance}, \eqref{eq:aklt_transition}, and \eqref{eq:aklt_emission}. Let $\varphi^{(\text{c})}$  with $(\text{c})\in\{ (\text{caus}), (\text{conv})\}$, be the global state on $\mathcal{A}_{\mathbb{N}_0}$ generated by $\Xi$ via the conventional causal structure \eqref{eq:conv_sliced} and \eqref{eq:caus_sliced}. Then for all $g \in G$, $A \in \mathcal{A}_{H,\mathbb{N}_0}$, and $B \in \mathcal{A}_{O,\mathbb{N}_0}$,
\begin{equation}
\varphi^{(\text{c})}\bigl( \alpha_g^{(H)}(A) \otimes \beta_g(B) \bigr) = \varphi^{(\text{c})}(A \otimes B),
\label{eq:aklt_global_invariance}
\end{equation}
where $\alpha_g^{(H)} = \operatorname{Ad}(\pi(g))$ and  $\beta_g = \operatorname{Ad}(\rho(g))$.
\end{theorem}

\begin{proof}
Lemma \ref{lem:aklt_emission_covariance} and Lemma \ref{lem:aklt_transition_equivariance} together verify the hypotheses of Theorem \ref{thm:global_invariance}. The result then follows directly.
\end{proof}

\begin{corollary}\label{cor:aklt_spt}
The global HQMM state $\varphi^{(\text{conv})}$ generated by the AKLT tensors carries a nontrivial SPT order characterized by the projective cocycle $[\omega] \in H^2(\mathrm{SO}(3),\mathrm{U}(1)) \cong \mathbb{Z}_2$. The emission map $\mathcal{E}_{O,H}$ acts as a cocycle intertwiner mediating between the projective hidden dynamics and the linear observable dynamics, thereby preserving the topological protection throughout the quantum Markov evolution.
\end{corollary}

Thus, when viewed as a stationary state of a causal HQMM, the AKLT state exhibits a robust SPT order encoded in the nontrivial cohomology class $[\omega]$. This hierarchical symmetry preservation—from local tensors through the emission map and hidden transition to the global state—reveals how SPT order emerges from the consistent interplay between hidden dynamics and symmetry constraints across all scales of the system \cite{Pollmann2010, Tas23}.

 \section{Discussion and Outlook}\label{Sect_Discussion}

The framework developed in this paper establishes a rigorous algebraic foundation for studying symmetry actions onHQMMs. By formulating the symmetry conditions \eqref{eq:initial_invariance}, \eqref{eq:Eh_equivariance}, and \eqref{eq:EOH_covariance}, we have shown that the projective nature of the hidden representation $\pi: G \to \mathcal{U}(\mathcal{H})$ with associated $2$-cocycle $\omega$ is precisely compensated by the emission map $\mathcal{E}_{O,H}: \mathcal{A}_H \otimes \mathcal{A}_O \to \mathcal{A}_H$, yielding global invariance of the HQMM state under the combined action $(\alpha^{(H)} \otimes \beta)(G)$. Theorem \ref{thm:global_invariance} provides a general criterion for symmetry protection applicable to any HQMM satisfying these conditions, independent of the causal structure.

From a mathematical perspective, the covariance condition \eqref{eq:EOH_covariance} defines a \textit{cocycle intertwiner} between the projective representation $\pi$ on $\mathcal{A}_H$ and the linear representation $\rho$ on $\mathcal{A}_O$, trivializing the obstruction $[\omega] \in H^2(G, \mathrm{U}(1))$ upon passage through the emission map. More generally, one may consider twisted equivariance conditions wherein the hidden transition map $\mathcal{E}_H$ carries its own $2$-cocycle, leading to a refined classification of HQMMs by higher cohomology groups—a perspective aligning with recent work on $C^*$-dynamical systems \cite{Soza19} and adapted unitary cocycles \cite{Vo25}.

From an applicative standpoint, the present framework bears direct relevance to quantum memory and quantum machine learning. HQMMs furnish a natural model for sequential quantum data with latent structure, and the symmetry conditions identified offer a principled means of incorporating physical symmetries into learning algorithms. In particular, the covariance condition ensures that the emission map respects the symmetry, a property essential for learning symmetry-invariant representations. Moreover, the stability of SPT order under quantum Markov evolution suggests that HQMMs equipped with nontrivial cocycles may serve as robust quantum memories.

Several avenues merit further investigation. First, the classification of covariant HQMMs for a given group $G$ and cocycle $[\omega]$ remains open, including the possibility of twisted equivariance. Second, extending the framework to hidden quantum Markov fields on $\mathbb{Z}^d$ would accommodate higher-dimensional causal structures, where the classification of SPT phases involves higher cohomology groups $H^{d+1}(G, \mathrm{U}(1))$ \cite{Chen2011b, Kallel2026, kitaev2006anyons}. Third, the mixing behavior and ergodicity of the two causal structures under symmetry constraints warrant systematic analysis using the tools of $C^*$-dynamical systems \cite{OP87}. Fourth, index-theoretic invariants along the lines of recent work on SPT phases \cite{O21, O23, Tas23} could provide numerical classifications of topological order within the HQMM setting. Finally, symmetry-aware learning algorithms leveraging the covariance conditions hold promise for applications in quantum state tomography and quantum process learning.

\appendix
\section{Appedix: Hidden Quantum Markov Models}\label{sect_prel}

Let $\mathcal{H}$ and $\mathcal{K}$ be separable Hilbert spaces, representing the hidden and observable degrees of freedom at a single lattice site, respectively. The corresponding single-site algebras are defined as the $C^*$-algebras of bounded linear operators: $\mathcal{A}_H := \mathcal{B}(\mathcal{H})$ and $\mathcal{A}_O := \mathcal{B}(\mathcal{K})$, each equipped with the operator norm and containing the identity operator. For each site $n \in \mathbb{Z}$, the full single-site algebra is the spatial $C^*$-tensor product
\[
\mathcal{A}_{\{n\}} := \mathcal{A}_H \otimes \mathcal{A}_O \cong \mathcal{B}(\mathcal{H} \otimes \mathcal{K}),
\]
where the isomorphism holds by separability of $\mathcal{H}$ and $\mathcal{K}$.

For a finite subset $\Lambda \subset \mathbb{Z}$, the finite-volume composite algebra is the $C^*$-tensor product over the sites:
\[
\mathcal{A}_\Lambda := \bigotimes_{n \in \Lambda} \mathcal{A}_{\{n\}} = \bigotimes_{n \in \Lambda} \big( \mathcal{A}_H \otimes \mathcal{A}_O \big)_n,
\]
where $(\mathcal{A}_H \otimes \mathcal{A}_O)_n$ denotes an isomorphic copy of the single-site algebra localized at position $n$. The finite-volume hidden and observable subalgebras are given respectively by
\[
\mathcal{A}_{H,\Lambda} := \bigotimes_{n \in \Lambda} (\mathcal{A}_H)_n, \qquad
\mathcal{A}_{O,\Lambda} := \bigotimes_{n \in \Lambda} (\mathcal{A}_O)_n,
\]
which satisfy the factorization $\mathcal{A}_\Lambda = \mathcal{A}_{H,\Lambda} \otimes \mathcal{A}_{O,\Lambda}$.

For each $n \in \mathbb{Z}$, define the hidden embedding $j_H^{(n)}: \mathcal{A}_H \hookrightarrow \mathcal{A}_{H,\mathbb{Z}}$ by
\[
j_H^{(n)}(X) := \bigotimes_{k \in \mathbb{Z}} X_k^{(H)}, \qquad
X_k^{(H)} = \begin{cases}
X & \text{if } k = n,\\
\mathbb{I}_{\mathcal{A}_H} & \text{otherwise},
\end{cases}
\]
extended continuously to the completion. This map is an isometric $*$-homomorphism. The observable embedding $j_O^{(n)}: \mathcal{A}_O \hookrightarrow \mathcal{A}_{O,\mathbb{Z}}$ is defined analogously as
\[
j_O^{(n)}(Y) := \bigotimes_{k \in \mathbb{Z}} Y_k^{(O)}, \qquad
Y_k^{(O)} = \begin{cases}
Y & \text{if } k = n,\\
\mathbb{I}_{\mathcal{A}_O} & \text{otherwise}.
\end{cases}
\]
The composite embedding at site $n$ is $j^{(n)} := j_H^{(n)} \otimes j_O^{(n)}$, satisfying $j^{(n)}(X \otimes Y) = j_H^{(n)}(X) \cdot j_O^{(n)}(Y)$. For any finite $\Lambda \subset \mathbb{Z}$, the embedding $j_\Lambda: \bigotimes_{n \in \Lambda} (\mathcal{A}_H \otimes \mathcal{A}_O) \hookrightarrow \mathcal{A}_{\mathbb{Z}}$ is defined on elementary tensors by
\[
j_\Lambda\Big( \bigotimes_{n \in \Lambda} (X_n \otimes Y_n) \Big) = \prod_{n \in \Lambda} j^{(n)}(X_n \otimes Y_n),
\]
with the product taken in increasing order of indices. These embeddings are consistent: for $\Lambda \subset \Lambda'$, we have $j_{\Lambda'} \circ \iota_{\Lambda,\Lambda'} = j_\Lambda$, where $\iota_{\Lambda,\Lambda'}$ is the inclusion map tensoring with identities on the additional sites.

The collection $\{\mathcal{A}_\Lambda, \iota_{\Lambda,\Lambda'}\}$ forms a directed system under inclusion. The local algebra is the algebraic direct limit $\mathcal{A}_{\mathbb{Z}}^{\text{loc}} := \bigcup_{\Lambda \subset \mathbb{Z} \text{ finite}} \mathcal{A}_\Lambda$, and the quasi-local $C^*$-algebra is its norm completion $\mathcal{A}_{\mathbb{Z}} := \overline{\mathcal{A}_{\mathbb{Z}}^{\text{loc}}}^{\|\cdot\|}$. When $\dim(\mathcal{H} \otimes \mathcal{K}) < \infty$, $\mathcal{A}_{\mathbb{Z}}$ is a uniformly hyperfinite (UHF) $C^*$-algebra. The hidden and observable quasi-local algebras $\mathcal{A}_{H,\mathbb{Z}}$ and $\mathcal{A}_{O,\mathbb{Z}}$ are constructed analogously, satisfying the spatial tensor product relation $\mathcal{A}_{\mathbb{Z}} \cong \mathcal{A}_{H,\mathbb{Z}} \otimes \mathcal{A}_{O,\mathbb{Z}}$.

The shift automorphism $\sigma \in \operatorname{Aut}(\mathcal{A}_{\mathbb{Z}})$ is defined on elementary tensors by $\sigma\big( \bigotimes_{n \in \mathbb{Z}} A_n \big) = \bigotimes_{n \in \mathbb{Z}} A_{n+1}$, extended continuously. It satisfies $\sigma \circ j^{(n)} = j^{(n+1)}$ for all $n \in \mathbb{Z}$, and $\{\sigma^n\}_{n \in \mathbb{Z}}$ forms a $\mathbb{Z}$-action on $\mathcal{A}_{\mathbb{Z}}$. There exists a strongly continuous one-parameter automorphism group $\tau: \mathbb{R} \to \operatorname{Aut}(\mathcal{A}_{\mathbb{Z}})$ with $\tau_1 = \sigma$, constructed via the crossed product $\mathcal{A}_{\mathbb{Z}} \rtimes_\sigma \mathbb{Z}$, with Liouvillian $H$ satisfying $\tau_t = \exp(i t H) (\cdot) \exp(-i t H)$. The triple $(\mathcal{A}_{\mathbb{Z}}, \mathbb{R}, \tau)$ is a $C^*$-dynamical system.

Restricting to the half-infinite chain indexed by $\mathbb{N}_0 = \{0,1,2,\dots\}$, the one-sided quasi-local algebra is
\[
\mathcal{A}_{\mathbb{N}_0} := \overline{\bigcup_{m \in \mathbb{N}_0} \bigotimes_{n=0}^m (\mathcal{A}_H \otimes \mathcal{A}_O)_n}^{\|\cdot\|},
\]
with the embeddings $j_H^{(n)}$, $j_O^{(n)}$, $j^{(n)}$ restricted to $n \ge 0$. The shift restricts to a non-unital endomorphism $\sigma_+: \mathcal{A}_{\mathbb{N}_0} \to \mathcal{A}_{\mathbb{N}_0}$ given by $\sigma_+(A) = \mathbb{I}_{\{0\}} \otimes A$ under the identification $\mathcal{A}_{\mathbb{N}_0} \cong \mathcal{A}_{\{0\}} \otimes \mathcal{A}_{\mathbb{N}_0}$.

For each $n \in \mathbb{N}_0$, define the finite-volume algebra $\mathcal{A}_{[0,n]} := \bigotimes_{k=0}^n (\mathcal{A}_H \otimes \mathcal{A}_O)_k$, a finite $C^*$-tensor product. For $m \le n$, the natural embedding $\iota_{m,n}: \mathcal{A}_{[0,m]} \hookrightarrow \mathcal{A}_{[0,n]}$ appends $n-m$ identity factors. The family $\{\mathcal{A}_{[0,n]}, \iota_{m,n}\}_{m \le n}$ forms an inductive system, and $\mathcal{A}_{\mathbb{N}_0}$ is the norm completion of its inductive limit \cite{OP87, Partha12}.

Let $\{\varphi_n\}_{n \in \mathbb{N}_0}$ be a sequence where each $\varphi_n$ is a state on $\mathcal{A}_{[0,n]}$. The sequence is said to satisfy the \textit{Kolmogorov compatibility condition} if for all $m \le n$,
\[
\varphi_n \circ \iota_{m,n} = \varphi_m.
\]
This condition guarantees consistency of expectation values across different volumes.

\begin{lemma}[Projective Limit State]
If $\{\varphi_n\}_{n \in \mathbb{N}_0}$ satisfies the Kolmogorov compatibility condition, then there exists a unique state $\varphi$ on $\mathcal{A}_{\mathbb{N}_0}$ such that $\varphi|_{\mathcal{A}_{[0,n]}} = \varphi_n$ for all $n$. Moreover, for any $A \in \mathcal{A}_{\mathbb{N}_0}$, $\varphi(A) = \lim_{n \to \infty} \varphi_n(A_n)$ whenever $A_n \in \mathcal{A}_{[0,n]}$ converges to $A$ in norm.
\end{lemma}

\begin{proof}
For any local observable $A$ belonging to some $\mathcal{A}_{[0,n]}$, set $\varphi(A) = \varphi_n(A)$. Compatibility makes this well-defined. This gives a linear functional on the dense subalgebra of local observables, which is positive and satisfies $\varphi(\mathbb{I}) = 1$. Hence it extends uniquely to a state on the whole algebra $\mathcal{A}_{\mathbb{N}_0}$ by continuity. The approximation property follows directly from the construction.
\end{proof}

The HQMM structure is encoded in a completely positive unital (CPU) map
\[
\mathcal{T}: \mathcal{A}_H \otimes \mathcal{A}_H \otimes \mathcal{A}_O \to \mathcal{A}_H,
\]
which governs the joint evolution of hidden and observable degrees of freedom across adjacent sites. For any $X \in \mathcal{A}_H$ and $Y \in \mathcal{A}_O$, define the sliced maps $\mathcal{T}_{X,Y}: \mathcal{A}_H \to \mathcal{A}_H$ by
\[
\mathcal{T}_{X,Y}(Z) := \mathcal{T}(X \otimes Z \otimes Y).
\]
Given an initial state $\phi_0$ on $\mathcal{A}_H$, we construct finite-volume states as follows. For each $n \in \mathbb{N}_0$, define $\varphi_n$ on $\mathcal{A}_{[0,n]}$ on elementary tensors by
\[
\varphi_n\!\left( \bigotimes_{k=0}^n (X_k \otimes Y_k) \right) := \phi_0 \circ \mathcal{T}_{X_0,Y_0} \circ \mathcal{T}_{X_1,Y_1} \circ \cdots \circ \mathcal{T}_{X_n,Y_n}(\mathbb{I}_{\mathcal{A}_H}),
\]
where the composition denotes successive application of the maps, and $\phi_0$ is evaluated on the resulting element of $\mathcal{A}_H$.

\begin{theorem}[Global HQMM State]
Let $\mathcal{T}: \mathcal{A}_H \otimes \mathcal{A}_H \otimes \mathcal{A}_O \to \mathcal{A}_H$ be a CPU map and $\phi_0$ a state on $\mathcal{A}_H$. For each $n \in \mathbb{N}_0$, the functional $\varphi_n$ defined above extends uniquely to a state on $\mathcal{A}_{[0,n]}$. The sequence $\{\varphi_n\}_{n \in \mathbb{N}_0}$ satisfies the Kolmogorov compatibility condition, and thus there exists a unique state $\varphi$ on $\mathcal{A}_{\mathbb{N}_0}$ such that $\varphi|_{\mathcal{A}_{[0,n]}} = \varphi_n$ for all $n$. This $\varphi$ is called the \textit{global HQMM state} generated by $(\phi_0, \mathcal{T})$.
\end{theorem}

\begin{proof}
Unitality and complete positivity of $\mathcal{T}$ imply that each $\mathcal{T}_{X,Y}$ is a CP map for fixed positive elements $X$ and $Y$,  and if $X =\mathbb{I}_{\mathcal{A}_H}$ and $Y= \mathbb{I}_{\mathcal{A}_O}$ the map $Z \mapsto \mathcal{T}_{\mathbb{I}_{\mathcal{A}_H},\mathbb{I}_{\mathcal{A}_H}}(Z) = \mathcal{T}(\mathbb{I}_{\mathcal{A}_H}\otimes Z\otimes \mathbb{I}_{\mathcal{A}_O})$ is CPU . Hence  for a given sequence of positive elements $X_i,Y_i, i =0,1,\cdots,n$, the map $\phi_0 \circ \mathcal{T}_{X_0,Y_0} \circ \cdots \circ \mathcal{T}_{X_n,Y_n}(\cdot)$ is a positive linear functional on $\mathcal{A}_H$. Moreover, $\varphi_n(\mathbb{I}) = 1$. By the universal property of the $C^*$-tensor product, $\varphi_n$ extends to a state on $\mathcal{A}_{[0,n]}$. For compatibility, consider $A \in \mathcal{A}_{[0,n]}$ and its embedding $\iota_{n,n+1}(A) = A \otimes (\mathbb{I}_{\mathcal{A}_H} \otimes \mathbb{I}_{\mathcal{A}_O})$. Then
\[
\varphi_{n+1}(\iota_{n,n+1}(A)) = \phi_0 \circ \mathcal{T}_{X_0,Y_0} \circ \cdots \circ \mathcal{T}_{X_n,Y_n} \circ \mathcal{T}_{\mathbb{I}_{\mathcal{A}_H},\mathbb{I}_{\mathcal{A}_O}}(\mathbb{I}_{\mathcal{A}_H}).
\]
Unitality of $\mathcal{T}$ gives $\mathcal{T}_{\mathbb{I}_{\mathcal{A}_H},\mathbb{I}_{\mathcal{A}_O}} = \mathbb{I}_{\mathcal{A}_H}$, so the right-hand side equals $\phi_0 \circ \mathcal{T}_{X_0,Y_0} \circ \cdots \circ \mathcal{T}_{X_n,Y_n}(\mathbb{I}_{\mathcal{A}_H}) = \varphi_n(A)$. By linearity and continuity, $\varphi_{n+1} \circ \iota_{n,n+1} = \varphi_n$, and similarly for all $m \le n$. The existence and uniqueness of $\varphi$ follow from Lemma 1.
\end{proof}

The hidden marginal state $\varphi_H$ on $\mathcal{A}_{H,\mathbb{N}_0}$ is defined by restriction: for $A \in \mathcal{A}_{H,\mathbb{N}_0}$, $\varphi_H(A) := \varphi(A \otimes \mathbb{I}_{\mathcal{A}_{O,\mathbb{N}_0}})$. For a local observable $\bigotimes_{k=0}^n X_k$, this becomes
\[
\varphi_H\!\left( \bigotimes_{k=0}^n X_k \right) = \phi_0 \circ \mathcal{T}_{X_0,\mathbb{I}_{\mathcal{A}_O}} \circ \cdots \circ \mathcal{T}_{X_n,\mathbb{I}_{\mathcal{A}_O}}(\mathbb{I}_{\mathcal{A}_H}),
\]
defining a quantum Markov chain on $\mathcal{A}_{H,\mathbb{N}_0}$. The observation marginal $\varphi_O$ on $\mathcal{A}_{O,\mathbb{N}_0}$ is defined analogously: for $B \in \mathcal{A}_{O,\mathbb{N}_0}$, $\varphi_O(B) := \varphi(\mathbb{I}_{\mathcal{A}_{H,\mathbb{N}_0}} \otimes B)$. For a local observable $\bigotimes_{k=0}^n Y_k$,
\[
\varphi_O\!\left( \bigotimes_{k=0}^n Y_k \right) = \phi_0 \circ \mathcal{T}_{\mathbb{I}_{\mathcal{A}_H},Y_0} \circ \cdots \circ \mathcal{T}_{\mathbb{I}_{\mathcal{A}_H},Y_n}(\mathbb{I}_{\mathcal{A}_H}),
\]
which represents the statistics of the observed process. The global state $\varphi$ serves as a coupling of these two marginals.

Two distinct causal structures arise as special realizations of $\mathcal{T}$, corresponding to different orderings of hidden transition and observable emission. Let $\mathcal{E}_H: \mathcal{A}_H \otimes \mathcal{A}_H \to \mathcal{A}_H$ be a CPU map describing the hidden transition, and $\mathcal{E}_{H,O}: \mathcal{A}_H \otimes \mathcal{A}_O \to \mathcal{A}_H$ a CPU map describing the emission. The \textit{conventional} causal structure is

\begin{equation}\label{Tconv}
\mathcal{T}^{(\text{conv})}(X \otimes X' \otimes Y) = \mathcal{E}_H\big( \mathcal{E}_{H,O}(X \otimes Y) \otimes X' \big),
\end{equation}
where emission acts first on the current hidden state $X$ and observable $Y$, followed by the transition that propagates the result together with the subsequent hidden state $X'$. The \textit{causal} causal structure is
\begin{equation}\label{Tcaus}
\mathcal{T}^{(\text{caus})}(X \otimes X' \otimes Y) = \mathcal{E}_{H,O}\big( \mathcal{E}_H(X \otimes X') \otimes Y \big).
\end{equation}
where the hidden transition acts first on the current hidden state $X$ and the subsequent hidden state $X'$, followed by emission coupling the evolved state with the observable $Y$. Both reduce to the same classical hidden Markov model when all algebras commute, but they yield distinct quantum correlation structures.

This $C^*$-dynamical framework provides the rigorous foundation for analyzing symmetry-protected topological order in HQMMs. The quasi-local algebra $\mathcal{A}_{\mathbb{N}_0}$, equipped with the shift endomorphism $\sigma_+$ and the embedding structure $j^{(n)}$, supports the definition of the generative triple $(\phi_0, \mathcal{E}_H, \mathcal{E}_{H,O})$ and the two causal structures $\mathcal{T}^{(\text{conv})}$ and $\mathcal{T}^{(\text{caus})}$ as particular instances.

\section*{Declarations}

\subsection*{Conflict of Interest}
The authors declare no conflict of interest.

\subsection*{Data Availability}
No data were generated or analyzed during this study.

\subsection*{CRediT authorship contribution statement}

Abdessatar Souissi: Conceptualization, Investigation, Resources, Methodology, Writing – original draft, Writing – review \& editing.

Abdessatar Barhoumi: Visualization, Project administration, Methodology, Resources, Project administration.

\subsection*{Data availability}

No data was used for the research described in the article.

\subsection*{Funding}
The authors received no funding for this work.


\begin{thebibliography}{99}


\bibitem{H83}
Haldane, F. D. M.:
Nonlinear field theory of large-spin Heisenberg antiferromagnets:
semiclassically quantized solitons of the one-dimensional easy-axis Néel state.
\textit{Physical Review Letters} \textbf{50}(15), 1153 (1983).

\bibitem{Pollmann2010}
Pollmann, F., Turner, A. M., Berg, E., Oshikawa, M.:
Entanglement spectrum of a topological phase in one dimension.
\textit{Physical Review B} \textbf{81}, 064439 (2010).

\bibitem{Pollmann2012}
Pollmann, F., Berg, E., Turner, A. M., Oshikawa, M.:
Symmetry protection of topological order in one-dimensional quantum spin systems.
\textit{Physical Review B} \textbf{85}, 075125 (2012).

\bibitem{AKLT1987}
Affleck, I., Kennedy, T., Lieb, E. H., Tasaki, H.:
Rigorous results on valence-bond ground states in antiferromagnets.
\textit{Physical Review Letters} \textbf{59}(7), 799 (1987).

\bibitem{Affleck1988}
Affleck, I., Kennedy, T., Lieb, E. H., Tasaki, H.:
Valence bond ground states in isotropic quantum antiferromagnets.
\textit{Communications in Mathematical Physics} \textbf{115}(3), 477–528 (1988).

\bibitem{O21}
Ogata, Y.:
A $\mathbb{Z}_2$-index of symmetry protected topological phases with time reversal symmetry for quantum spin chains.
\textit{Communications in Mathematical Physics} \textbf{374}(2), 705–734 (2020).

\bibitem{O23}
Ogata, Y.:
Classification of gapped ground state phases in quantum spin systems.
In: \textit{International Congress of Mathematicians}, pp. 4142–4161.
European Mathematical Society, Helsinki (2023).

\bibitem{Tas23}
Tasaki, H.:
Rigorous index theory for one-dimensional interacting topological insulators.
\textit{Journal of Mathematical Physics} \textbf{64}, 041903 (2023).

\bibitem{Chen2011a}
Chen, X., Gu, Z.-C., Wen, X.-G.:
Classification of gapped symmetric phases in one-dimensional spin systems.
\textit{Physical Review B} \textbf{83}, 035107 (2011).

\bibitem{Chen2011b}
Chen, X., Gu, Z.-C., Wen, X.-G.:
Complete classification of one-dimensional gapped quantum phases in interacting spin systems.
\textit{Physical Review B} \textbf{84}, 235128 (2011).

\bibitem{D17}
Stephen, D. T., Wang, D.-S., Prakash, A., Wei, T.-C., Raussendorf, R.:
Computational power of symmetry-protected topological phases.
\textit{Physical Review Letters} \textbf{119}, 010504 (2017).

\bibitem{GTS22}
de Groot, C., Turzillo, A., Schuch, N.:
Symmetry protected topological order in open quantum systems.
\textit{Quantum} \textbf{6}, 856 (2022).

\bibitem{D20}
de Groot, C., Stephen, D. T., Molnar, A., Schuch, N.:
Inaccessible entanglement in symmetry protected topological phases.
\textit{Journal of Physics A: Mathematical and Theoretical} \textbf{53}(33), 335302 (2020).

\bibitem{Sati26}
Sati, H., Schreiber, U.:
Identifying anyonic topological order in fractional quantum anomalous Hall systems.
\textit{Applied Physics Letters} \textbf{128}(2) (2026).

\bibitem{AGLS24Q}
Accardi, L., Soueidi, E. G., Lu, Y. G., Souissi, A.:
Hidden quantum Markov processes.
\textit{Infinite Dimensional Analysis, Quantum Probability and Related Topics}
\textbf{28}(04), 2450007 (2025).

\bibitem{AGLS24C}
Accardi, L., Soueidi, E. G., Lu, Y. G., Souissi, A.:
Algebraic hidden processes and hidden Markov processes.
\textit{Infinite Dimensional Analysis, Quantum Probability and Related Topics}
\textbf{29}(01), 2450009 (2026).

\bibitem{Partha12}
Parthasarathy, K. R.:
\textit{An Introduction to Quantum Stochastic Calculus}.
Springer, Berlin (2012).

\bibitem{choi1975completely}
Choi, M.-D.:
Completely positive linear maps on complex matrices.
\textit{Linear Algebra and Its Applications} \textbf{10}(3), 285–290 (1975).

\bibitem{OP87}
Bratteli, O., Robinson, D. W.:
\textit{Operator Algebras and Quantum Statistical Mechanics}, Vols. 1 and 2, 2nd edn.
Springer, Berlin (1987).

\bibitem{ET}
Taranto, P., Elliott, T. J., Milz, S.:
Hidden quantum memory: Is memory there when somebody looks?
\textit{Quantum} \textbf{7}, 991 (2023).

\bibitem{LiXi24}
Li, X.-Y., et al.:
A new quantum machine learning algorithm: split hidden quantum Markov model inspired by quantum conditional master equation.
\textit{Quantum} \textbf{8}, 1232 (2024).

\bibitem{Vo25}
Voigt, F. R.:
Adapted unitary cocycles and quantum Markov processes: a cohomological approach.
Doctoral dissertation (2025).

\bibitem{fannes1992finitely}
Fannes, M., Nachtergaele, B., Werner, R. F.:
Finitely correlated states of quantum spin chains.
\textit{Communications in Mathematical Physics} \textbf{144}, 443–490 (1992).

\bibitem{fannes2}
Fannes, M., Nachtergaele, B., Werner, R. F.:
Ground states of VBS models on Cayley trees.
\textit{Journal of Statistical Physics} \textbf{66}, 939–973 (1992).

\bibitem{SouAndhqmm2025}
Souissi, A., Andolsi, A.:
A hidden quantum Markov model framework for entanglement and topological order in the AKLT chain.
\textit{The European Physical Journal Plus} \textbf{141}(4), 404 (2026).

\bibitem{Ragone24}
Ragone, M.:
SO($n$) AKLT chains as symmetry protected topological quantum ground states.
Ph.D. thesis, University of California, Davis (2024).

\bibitem{Xi18}
Xiong, C. Z.:
Minimalist approach to the classification of symmetry protected topological phases.
\textit{Journal of Physics A: Mathematical and Theoretical} \textbf{51}(44), 445001 (2018).

\bibitem{Soza19}
Szabó, G.:
On a categorical framework for classifying C$^*$-dynamics up to cocycle conjugacy.
\textit{Journal of Functional Analysis} \textbf{280}, 108927 (2019).

\bibitem{SouBarhqmm2026}
Souissi, A., Barhoumi, A.:
Causal architecture in hidden quantum Markov models.
arXiv:2602.19120 (2026).


\bibitem{smith19}
Smith, A., Kim, M. S., Pollmann, F., Knolle, J.:
Simulating quantum many-body dynamics on a current digital quantum computer.
\textit{npj Quantum Information} \textbf{5}, 106 (2019).

\bibitem{Kallel2026}
Kallel, S., Sati, H., Schreiber, U.:
Higher-Dimensional Anyons via Higher Cohomotopy.
arXiv:2601.03150 (2026).

\bibitem{kitaev2006anyons}
Kitaev, A.:
Anyons in an exactly solved model and beyond.
\textit{Annals of Physics} \textbf{321}(1), 2–111 (2006).

\end{thebibliography}
\end{document}